\documentclass[11pt]{article}
\usepackage[letterpaper,hmargin=1in,vmargin=1.25in]{geometry}
\usepackage[english]{babel}
\usepackage[utf8]{inputenc}
\usepackage[T1]{fontenc}
\usepackage{amsmath}
\usepackage{amsfonts}
\usepackage{amssymb}
\usepackage{enumitem}
\usepackage{framed}
\usepackage{dsfont} 
\usepackage[standard]{ntheorem}
\usepackage[dvips]{graphics}
\usepackage{epsfig}
\usepackage{epsfig,psfrag}
\usepackage{subfigure}
\usepackage{color}
\usepackage{calc}
\usepackage{listings}
\usepackage{url}
\usepackage{makeidx}
\usepackage{stmaryrd}
\usepackage{alltt}
\usepackage{tabls}
\usepackage{textcomp}
\usepackage{slashbox}
\usepackage{times}

\makeindex

\title{A topological chaos framework for hash functions}

\begin{document}

\title{A topological chaos framework for hash functions}
\date{ }
\author{}
\maketitle

\begin{abstract}
This paper presents a new procedure of generating hash functions which can
be evaluated using some mathematical tools. This procedure is based on
discrete chaotic iterations.

First, it is mathematically proven, that these discrete chaotic iterations
can be considered as a \linebreak particular case of topological chaos. Then, the
process of generating hash function based on the \linebreak topological chaos  is
detailed. Finally it is shown how some tools coming from the domain of
\linebreak topological chaos can be used to measure quantitatively and qualitatively
some desirable properties for hash functions. An illustration example is
detailed in order to show how one can create hash functions using our
theoretical study.
\end{abstract}

\setcounter{secnumdepth}{3}


\section{Introduction}

Hash functions, such as MD5 or SHA-256, can be described by discrete
iterations on a finite set. In this paper, the elements of this finite set
are called cells. These cells represent the blocks of the text to which the
hash function will be applied. The origin of this study goes up with the
idea of using the concept of discrete \textit{chaotic} iterations for
generating new hash functions. This idea gave then quickly rise to the
question of knowing if discrete chaotic iterations really generate chaos.

This article presents the research results related to this question.%
\newline
First, we prove that under some conditions, discrete chaotic iterations
produce chaos, precisely, they produce topological chaos in the sense of
Devaney. This topological chaos is a rigorous
framework well studied in the field of mathematical theory of chaos.
Thanks to this result we give a process of generating hash functions.

Behind the theoretical interest connecting the field of the chaotic discrete
iterations and the one of topological chaos, our study gives a framework
making it possible to create hash functions that can be mathematically
evaluated and compared.\newline
Indeed, some required qualities for hash functions such as strong
sensitivity to the original text, resistance to collisions and unpredictability
 can be mathematically described by notions
from the theory of topological chaos, namely, sensitivity, transitivity,
entropy and expansivity. These concepts are approached but non deepened in
this article. More detailed studies will be carried out in forthcoming
articles.

This study is the first of a series we intend to carry out. We think that
the mathematical framework in which we are placed offers interesting new
tools allowing the conception, the comparison and the evaluation of new
methods of encryption in general, not only hash functions.

The rest of the paper is organized as follows.\newline
The first next section is devoted to some recalls on two distinct domains, the
domain of topological chaos and the domain of discrete chaotic iterations.%
\newline
Third and fourth sections constitute the theoretical study of the present
paper. Section III defines the \linebreak topological framework in which we are placed
while section IV shows that the chaotic iterations produce a topological
chaos.\newline
The following section details, using an illustration example, the procedure
to build hash functions based on our theoretical results. Section VI explains
how quantitative measures could be obtained for hash functions. The paper ends
by some discussions and future work.

\section{Basic recalls}

This section is devoted to basic definitions and terminologies in the field
of topological chaos and in the one of chaotic iterations.

\subsection{Devaney's chaotic dynamical systems}

Consider a metric space $(\mathcal{X},d)$, and a continuous function $f:%
\mathcal{X}\longrightarrow \mathcal{X}$.

\begin{definition}
$f$ is said to be \emph{topologically transitive} if, for any pair of open
sets $U,V \subset \mathcal{X}$, there exists $k>0$ such that $f^k(U) \cap V
\neq \varnothing$.
\end{definition}

\begin{definition}
An element (a point) $x$ is a periodic element (point) for $f$ of period $%
n\in \mathds{N},$ if $f^{n}(x)=x$. The set of periodic points of $f$ is denoted $%
Per(f).$
\end{definition}

\begin{definition}
$(\mathcal{X},f)$ is said to be \emph{regular} if the set of periodic points
is dense in $\mathcal{X}$, 
\begin{equation*}
\forall x\in \mathcal{X},\forall \varepsilon >0,\exists p\in
Per(f),d(x,p)\leqslant \varepsilon .
\end{equation*}
\end{definition}

\begin{definition}
\label{sensitivity} $f$ has \emph{sensitive dependence on initial conditions}
if there exists $\delta >0$ such that, for any $x\in \mathcal{X}$ and any
neighborhood $V$ of $x$, there exists $y\in V$ and $n\geqslant 0$ such that $%
|f^{n}(x)-f^{n}(y)|>\delta $.

$\delta$ is called the \emph{constant of sensitivity} of $f$.
\end{definition}

Let us now recall the definition of a chaotic topological system, in the
sense of Devaney~\cite{Dev89} :

\begin{definition}
$f:\mathcal{X}\longrightarrow \mathcal{X}$ is said to be \emph{chaotic} on $%
\mathcal{X}$ if,

\begin{enumerate}
\item $f$ has sensitive dependence on initial conditions,

\item $f$ is topologically transitive,

\item $(\mathcal{X},f)$ is regular.
\end{enumerate}
\end{definition}

Therefore, quoting Robert Devaney: ``A chaotic map possesses three
ingredients: unpredictability, indecomposability, and an element of
regularity. A chaotic system is unpredictable because of the sensitive
dependence on initial conditions. It cannot be broken down or decomposed
into two subsystems, because of topological transitivity. And, in the midst
of this random behavior, we nevertheless have an element of regularity,
namely the periodic points which are dense.''

Banks \emph{et al.} proved in ~\cite{Banks92} that sensitive dependence is a
consequence of being regular and topologically transitive.

\subsection{Chaotic iterations}

In the sequel $s[n]$ denotes the $n-$th term of a sequence $s$, $V_{i}$
denotes the $i-$th component of a vector $V$, and $f^{k}$ denotes the $k-$th
composition of a function $f$. Finally, the following notation is used: $%
\llbracket1;N\rrbracket=\{1,2,\hdots,N\}$.

\medskip

Let us consider a \emph{system} of a finite number $\mathsf{N}$ of \emph{%
cells} so that each cell has a boolean \emph{state}. Then a sequence of
length $\mathsf{N}$ of boolean states of the cells corresponds to a
particular \emph{state of the system}.

A \emph{strategy} corresponds to a sequence of $\llbracket1;\mathsf{N}%
\rrbracket$. The set of all strategies is denoted by $\mathcal{S}.$

\begin{definition}
Let $S\in \mathcal{S}$. The \emph{shift} function is defined by%
\begin{equation*}
\begin{array}{lclc}
\sigma : & \mathcal{S} & \longrightarrow & \mathcal{S} \\ 
& (S[n])_{n\in \mathds{N}} & \longmapsto & (S[n+1])_{n\in \mathds{N}}%
\end{array}%
\end{equation*}%
\noindent and the \emph{initial function} is the map which associates to a
sequence, its first term 
\begin{equation*}
\begin{array}{lclc}
i: & \mathcal{S} & \longrightarrow & \llbracket1;\mathsf{N}\rrbracket \\ 
& (S[n])_{n\in \mathds{N}} & \longmapsto & S[0].%
\end{array}%
\end{equation*}
\end{definition}

$\mathds{B}$ denoting $\{0,1\}$, let $f:\mathds{B}^{\mathsf{N}}\longrightarrow \mathds{B}^{\mathsf{N}}$ and $%
S\in \mathcal{S}$ be a strategy. Let us consider the following so called 
\emph{chaotic iterations} (see \cite{Robert1986} for the general definition
of such iterations).

\begin{equation}
\left\{ 
\begin{array}{l}
x[0]\in \mathds{B}^{\mathsf{N}} \\ 
\forall n\in \mathds{N}^{\ast },\forall i\in \llbracket1;\mathsf{N}\rrbracket%
,x[n]_{i}=\left\{ 
\begin{array}{ll}
x[n-1]_{i} & \text{ if }S[n]\neq i \\ 
f(x[n])_{S[n]} & \text{ if }S[n]=i.%
\end{array}%
\right.%
\end{array}%
\right.  \label{chaotic iterations}
\end{equation}

In other words, at the $n-$th iteration, only the $S[n]-$th cell is
\textquotedblleft iterated\textquotedblright . Note that in a more general
formulation, $f(x[n])_{S[n]}$ can be replaced by $f(x[k])_{S[n]}$, where $%
k\leqslant n$, modelizing for example delay transmission (see \emph{e.g.}  \cite{Bahi2000}).

\section{A topological approach of chaotic iterations}

\subsection{The new topological space}

In this section we will put our study in a topological context by defining a
suitable set and a suitable distance.

\subsubsection{Defining the iteration function and the phase space}

\label{Defining}

Let us denote by $\delta $ the \emph{discrete boolean metric}, $\delta
(x,y)=0\Leftrightarrow x=y,$ and define the function

\begin{equation*}
\begin{array}{lrll}
F_{f}: & \llbracket1;\mathsf{N}\rrbracket\times \mathds{B}^{\mathsf{N}} & 
\longrightarrow & \mathds{B}^{\mathsf{N}} \\ 
& (k,E) & \longmapsto & \left( E_{j}.\delta (k,j)+f(E)_{k}.\overline{\delta
(k,j)}\right) _{j\in \llbracket1;\mathsf{N}\rrbracket},%
\end{array}%
\end{equation*}%
where + and . are boolean operations.

\bigskip

Consider the phase space%
\begin{equation*}
\mathcal{X} = \llbracket 1 ; \mathsf{N} \rrbracket^\mathds{N} \times
\mathds{B}^\mathsf{N},
\end{equation*}

\noindent and the map 
\begin{equation}
G_f\left(S,E\right) = \left(\sigma(S), F_f(i(S),E)\right)   \label{Gf}
\end{equation}

\noindent Then one can remark that the chaotic iterations defined in (\ref%
{chaotic iterations}) can be described by the following iterations 
\begin{equation*}
\left\{ 
\begin{array}{l}
X[0]\in \mathcal{X} \\ 
X[k+1]=G_{f}(X[k]).%
\end{array}%
\right.
\end{equation*}%
\bigskip

The following result can be easily proven, by comparing $\mathcal{S}$ and $\mathds{R}$ that,

\begin{theorem}
The phase space $\mathcal{X}$ has the cardinality of the continuum.
\end{theorem}

Note that this result is independent on the number of cells.

\subsubsection{A new distance}

We define a new distance between two points $(S,E),(\check{S},\check{E})\in 
\mathcal{X}$ by%
\begin{equation*}
d((S,E);(\check{S},\check{E}))=d_{e}(E,\check{E})+d_{s}(S,\check{S}),
\end{equation*}

\noindent where

\begin{equation*}
\left\{ 
\begin{array}{lll}
\displaystyle{d_{e}(E,\check{E})} & = & \displaystyle{\sum_{k=1}^{\mathsf{N}%
}\delta (E_{k},\check{E}_{k})}, \\ 
\displaystyle{d_{s}(S,\check{S})} & = & \displaystyle{\dfrac{9}{\mathsf{N}}%
\sum_{k=1}^{\infty }\dfrac{|S[k]-\check{S}[k]|}{10^{k}}}.%
\end{array}%
\right. 
\end{equation*}

It should be noticed that if the floor function $\lfloor d(X,Y)\rfloor =n$,
then the strategies $X$ and $Y$ differs in $n$ cells and that $d(X,Y) -
\lfloor d(X,Y) \rfloor $ gives a measure on how the strategies $S$ and $%
\text{\v{S}}$ diverge. More precisely,

\begin{itemize}
\item This floating part is less than $10^{-k}$ if and only if the first $k$
terms of the two strategies are equal.

\item If the $k-$th digit is nonzero, then the $k-$th terms of the two
strategies are different.
\end{itemize}

\subsection{Continuity of the iteration function}

To prove that chaotic iterations are an example of topological chaos in the
sense of Devaney ~\cite{Dev89}, $G_{f}$ should be continuous on the metric
space $(\mathcal{X},d)$.

\begin{theorem}
$G_f$ is a continuous function.
\end{theorem}

\begin{proof}
We use the sequential continuity.

Let $(S[n],E[n])_{n\in \mathds{N}}$ be a sequence of the phase space $%
\mathcal{X}$, which converges to $(S,E)$. We will prove that $\left(
G_{f}(S[n],E[n])\right) _{n\in \mathds{N}}$ converges to $\left(
G_{f}(S,E)\right) $. Let us recall that for all $n$, $S[n]$ is a strategy,
thus, we consider a sequence of strategy (\emph{i.e.} a sequence of
sequences).\newline
As 
\begin{equation*}
d((S[n],E[n]);(S,E))
\end{equation*}%
converges to 0, each distance $d_{e}(E[n],E)$ and $d_{s}(S[n],S)$ converges
to 0. But $d_{e}(E[n],E)$ is an integer, so $\exists n_{0}\in \mathds{N},$ $%
d_{e}(E[n],E)=0$ for any $n\geqslant n_{0}$.\newline
In other words, there exists threshold $n_{0}\in \mathds{N}$ after which no
cell will change its state: 
\begin{equation*}
\exists n_{0}\in \mathds{N},n\geqslant n_{0}\Longrightarrow E[n]=E.
\end{equation*}%
In addition, $d_{s}(S[n],S)\longrightarrow 0,$ so $\exists n_{1}\in %
\mathds{N},d_{s}(S[n],S)<10^{-1}$ for all indices greater than or equal to $%
n_{1}$. This means that for $n\geqslant n_{1}$, all the $S[n]$ have the same
first term, which is $S[0]$:%
\begin{equation*}
\forall n\geqslant n_{1},S[n][0]=S[0].
\end{equation*}%
Thus, after the $max(n_{0},n_{1})-$th term, states of $E[n]$ and $E$ are the
same, and strategies $S[n]$ and $S$ start with the same first term.\newline
Consequently, states of $G_{f}(S[n],E[n])$ and $G_{f}(S,E)$ are equal, then
distance $d$ between this two points is strictly less than 1 (after the rank 
$max(n_{0},n_{1})$).\bigskip \newline
\noindent We now prove that the distance between $\left(
G_{f}(S[n],E[n])\right) $ and $\left( G_{f}(S,E)\right) $ is convergent to
0. Let $\varepsilon >0$. \medskip

\begin{itemize}
\item If $\varepsilon \geqslant 1$, then we have seen that the distance
between $\left( G_{f}(S[n],E[n])\right) $ and $\left( G_{f}(S,E)\right) $ is
strictly less than 1 after the $max(n_{0},n_{1})$-th term (same state).
\medskip

\item If $\varepsilon <1$, then $\exists k\in \mathds{N},10^{-k}\geqslant
\varepsilon \geqslant 10^{-(k+1)}$. But $d_{s}(S[n],S)$ converges to 0, so 
\begin{equation*}
\exists n_{2}\in \mathds{N},\forall n\geqslant
n_{2},d_{s}(S[n],S)<10^{-(k+2)},
\end{equation*}%
after $n_{2}$, the $k+2$-th first terms of $S[n]$ and $S$ are equal.
\end{itemize}

\noindent As a consequence, the $k+1$ first entries of the strategies of $%
G_{f}(S[n],E[n])$ and $G_{f}(S,E)$ are the same (because $G_{f}$ is a shift
of strategies), and due to the definition of $d_{s}$, the floating part of
the distance between $(S[n],E[n])$ and $(S,E)$ is strictly less than $%
10^{-(k+1)}\leqslant \varepsilon $.\bigskip \newline
In conclusion, $G_{f}$ is continuous,%
\begin{equation*}
\forall \varepsilon >0,\exists N_{0}=max(n_{0},n_{1},n_{2})\in \mathds{N}%
,\forall n\geqslant N_{0},d\left( G_{f}(S[n],E[n]);G_{f}(S,E)\right)
\leqslant \varepsilon .
\end{equation*}
\end{proof}

In this section, we proved that chaotic iterations can be modelized as a
dynamical system in a topological space. In the next section, we show that
chaotic iterations are a case of topological chaos, in the sense of Devaney.

\section{Discrete chaotic iterations are topological chaos}

To prove that we are in the framework of Devaney's topological chaos, we
will check the regularity and transitivity conditions.

\subsection{Regularity}

\begin{theorem}
Periodic points of $G_{f}$ are dense in $\mathcal{X}$.
\end{theorem}

\begin{proof}
Let $(S,E)\in \mathcal{X}$, and $\varepsilon >0$. We are looking for a
periodic point $(S^{\prime },E^{\prime })$ satisfying
\begin{equation*}
d((S,E);(S^{\prime },E^{\prime }))<\varepsilon .
\end{equation*}%
We choose $E^{\prime }=E$, and we \textquotedblleft copy\textquotedblright\
enough entries from $S$ to $S^{\prime }$ so that the distance between $%
(S^{\prime },E)$ and $(S,E)$ is strictly less than $\varepsilon $: a number $%
k=\lfloor log_{10}(\varepsilon )\rfloor +1$ of terms is sufficient.\newline
After this $k$-th iterations, the new common state is $\mathcal{E}$, and
strategy $S$ is shifted of $k$ positions: $\sigma ^{k}(S)$.\newline
Then we have to complete strategy $S^{\prime }$ in order to make $(E^{\prime
},S^{\prime })$ periodic (at least for sufficiently large indices). To do
so, we put an infinite number of 1 to the strategy $S^{\prime }$. Then,
either:

\begin{enumerate}
\item The first state is conserved after one iteration, so $\mathcal{E}$ is
unchanged and we obtain a fixed point. Or

\item The first state is not conserved, then:

\begin{itemize}
\item If the first state is not conserved after a second iteration, then we
will be again in the first case above (due to the negation function).

\item Otherwise the first state is conserved, and we have indeed a fixed
(periodic) point.
\end{itemize}
\end{enumerate}

Thus, there exists a periodic point into every neighborhood of any point, so 
$(\mathcal{X},G)$ is regular.
\end{proof}

\subsection{Transitivity}

Contrary to the regularity, the topological transitivity condition is not
automatically satisfied by any function ($f=Identity$ is not topologically
transitive).

Let us denote by $\mathcal{T}$ the set of maps $f$ such that $(\mathcal{X}%
,G_{f})$ is topologically transitive. Then.

\begin{theorem}
$\mathcal{T}$ is a nonempty set.
\end{theorem}

\begin{proof}
We will prove that the vectorial logical negation function $f_{0}$

\begin{equation}
\begin{array}{rccc}
f_{0}: & \mathds{B}^{\mathsf{N}} & \longrightarrow & \mathds{B}^{\mathsf{N}}
\\ 
& (x_{1},\hdots,x_{\mathsf{N}}) & \longmapsto & (\overline{x_{1}},\hdots,%
\overline{x_{\mathsf{N}}}) \\ 
&  &  & 
\end{array}
\label{f0}
\end{equation}%
\noindent is topologically transitive.\newline
Let $A=\mathcal{B}(X_{A},r_{A})$ and $B=\mathcal{B}(X_{B},r_{B})$ be two
open balls of $\mathcal{X}$. Our goal is to start from a point of $A$ (\emph{%
i.e.} a point close to $X_{A}$) and to arrive in $B$ (a point close to $%
X_{B} $).\newline
We have to be close to $X_{A}$, then the starting state is $E_{A}$, it
remains to determine the strategy $S$. We start by filling $S$ with the $%
n_{0}$ first terms of strategy $S_{A}$ of $X_{A}$, so that $(S,E_{A})\in
B_{A}$.\newline
Let $E$ be the image of the state $E_{A}$ by mapping the $n_{0}$-th first
terms of the strategy $S$. This new state $E$ differs from $E_{B}$ by a
finite number of states, we put these cells to our strategy $S$ (this adds $%
n_{1}$ integers to $S$).\newline
In short, starting from $(S,E_{A})$, we are in $X_{B}$ after $n_{0}+n_{1}$
iterations, and the strategy $S$ was shifted of $n_{0}+n_{1}$ terms (there
is no more term in $S$).\newline
In order to be sufficiently close to $(S_{B},E_{B})$ (at a distance less
than $\varepsilon $ from $(S_{B},E_{B})$), we add as much as necessarily
terms of $S_{B}$ to $S$ and we complete $S$ with an infinity of terms equal
to $1.$\newline
\end{proof}

\begin{remark}
In fact, we can prove that $(\mathcal{X},G_{f_{0}})$ is \emph{highly
topologically transitive} in the following sense: for every $(S_{A},E_{A})$
and $(S_{B},E_{B})$ of $\mathcal{X}$, there exists a point sufficiently
close to $(S_{A},E_{A})$ and $n_{0}\in \mathds{N}$ such that $%
G_{f_{0}}^{n_{0}}(S_{A},E_{A})=(S_{B},E_{B})$.
\end{remark}

In conclusion, if $f\in \mathcal{T}\neq \varnothing $, then $(\mathcal{X}%
,G_{f})$ is topologically transitive and regular, and then we have the
result.

\begin{theorem}
$\forall f \in \mathcal{T}, (\mathcal{X},G_f)$ is chaotic, in the sense of
Devaney.
\end{theorem}

\section{Hash functions based on topological chaos}

\subsection{Objective}

As an application of the previous theory, we define in this section a new
way to generate hash functions based on topological chaos. Our approach
guarantees to obtain various desired properties in the domain of encryption.
For example, the avalanche criterion is closely linked to the expansivity
property (see the next section below).

The following hash function is based on the vectorial boolean negation $f_{0}
$ defined in (\ref{f0}). Nevertheless, our procedure remains general, and
can be applied with any transitive function $f$.

\subsection{Application of the new hash function}

Our initial condition $X0=\left( S,E\right) $ is composed by:

\begin{itemize}
\item A 256 bits sequence that we call $E$, obtained from the original text.

\item A chaotic strategy $S$.
\end{itemize}

In the sequel, we describe how to obtain this initial condition $(S,E)$.

\subsubsection{How to obtain $E$}

The first step of our algorithm is to transform the message in a normalized
256 bits sequence $E$. To illustrate this step, we take an example, our
original text is: \emph{The original text}

Each character of this string is replaced by its ASCII code (on 7 bits).
Then, we add a 1 to this string.

\bigskip

\begin{center}
\begin{alltt}
\noindent 10101001 10100011 00101010 00001101 11111100 10110100
\noindent 11100111 11010011 10111011 00001110 11000100 00011101
\noindent 00110010 11111000 11101001 
\end{alltt}
\end{center}

\bigskip

Then, we add the binary value of the length of this string, and we add 1 one
more time:

\bigskip

\begin{center}
\begin{alltt}
\noindent 10101001 10100011 00101010 00001101 11111100 10110100
\noindent 11100111 11010011 10111011 00001110 11000100 00011101
\noindent 00110010 11111000 11101001 11110001
\end{alltt}
\end{center}

\bigskip

Then, the whole string is copied, but in the opposite direction, this gives:

\bigskip

\begin{center}
\begin{alltt}
\noindent 10101001 10100011 00101010 00001101 11111100 10110100
\noindent 11100111 11010011 10111011 00001110 11000100 00011101
\noindent 00110010 11111000 11101001 11110001 00011111 00101110
\noindent 00111110 10011001 01110000 01000110 11100001 10111011
\noindent 10010111 11001110 01011010 01111111 01100000 10101001
\noindent 10001011 0010101
\end{alltt}
\end{center}

\medskip So, we obtain a multiple of 512, by duplicating enough this 
string and truncating at a multiple of 512.
This string, in which contains the whole original text is denoted by $D$.

\bigskip

Finally, we split our obtained string into blocks of 256 bits, and apply to
them the exclusive-or function, obtaining a 256 bits sequence.

\bigskip 
\begin{alltt}
\noindent 11111010 11100101 01111110 00010110 00000101 11011101
\noindent 00101000 01110100 11001101 00010011 01001100 00100111
\noindent 01010111 00001001 00111010 00010011 00100001 01110010
\noindent 01000011 10101011 10010000 11001011 00100010 11001100
\noindent 10111000 01010010 11101110 10000001 10100001 11111010
\noindent 10011101 01111101 
\end{alltt}

So, in the context of subsection (1)
, $\mathsf{N}=256$, and $E$ is the above obtained sequence of 256 bits. Let
us now build the strategy $S$.

We now have the definitive length of our digest. Note that a lot of texts
have the same string. This is not a problem because the strategy we will
build will depends on the whole text.

\subsubsection{How to choose $S$}

In order to forge our strategy, \emph{i.e.} the sequence $S$ of $X[0]=(S,E)$%
, we use the previously obtained string $D$, and then we start by
constructing an intermediate sequence as follows:

\begin{enumerate}
\item We split this string into blocks of 8 bits, and we add to our sequence
the corresponding decimal value of each octet.

\item We take then the first bit of this string, and put it on the end. Then
we split the new string into blocks of 8 bits, and we add in the sequence
decimal value associated.

\item We repeat this operation 6 times.
\end{enumerate}

The general term of this sequence will be denoted by $(u[n])_{n}$.

\bigskip

Now, we are able to build our strategy $S$. The first term of $S$ is the
initial term of the preceding sequence. The $n-$th term is the sum (modulo
256) of the three following terms:

\begin{itemize}
\item the $n-$th term of the intermediate sequence (the strategy depends on
the original text),

\item the double of the $n-1$-th term of the strategy (introduction of
sensitivity, with the analogy with the well known chaotic map $\theta
\longmapsto 2\theta ~(mod ~1)$),

\item $n$ (to prevent periodic behaviour).
\end{itemize}

So, the general term $S[n]$ of $S$ is defined by%
\begin{equation*}
S[n]=(u[n]+2\times S[n-1]+n) ~(mod ~256).
\end{equation*}%
Strategy $S$ is strong sensitive to the modification of the original
text, because the map $\theta \longmapsto 2\theta ~(mod ~1)$ is known to be
chaotic in the sense of Devaney.

\subsubsection{How to construct the digest}

We apply the logical negation function to the $S[k]$-th term of $E$, (modulo
256). Indeed, the function $f$ of equation (\ref{f0}) is defined by 
\begin{equation*}
\begin{array}{rccc}
f: & \llbracket1,256\rrbracket & \longrightarrow  & \llbracket1,256\rrbracket
\\ 
& (E[1],\hdots,E[256]) & \longmapsto  & (\overline{E[1]},\hdots,\overline{%
E[256]}).%
\end{array}%
\end{equation*}%
It is possible to apply the logical negation function several times the same
bit.

\bigskip

\noindent We finally split these 256 bits into blocks of 4 bits, this will
returns the hexadecimal value:

\bigskip 
\begin{alltt}
\noindent 63A88CB6AF0B18E3BE828F9BDA4596A6A13DFE38440AB9557DA1C0C6B1EDBDBD
\end{alltt}

\bigskip

As a comparison if instead of considering the text \textquotedblleft \textit{%
The original text}\textquotedblright\ we took \textquotedblleft \textit{the original text}\textquotedblright , the hash function returns:

\bigskip 
\begin{alltt}
\noindent 33E0DFB5BB1D88C924D2AF80B14FF5A7B1A3DEF9D0E831194BD814C8A3B948B3
\end{alltt}

\subsection{Example}

Consider the following message (a \emph{E. A. Poe}'s poem):
\begin{center}
\begin{lstlisting}
                    Wanderers in that happy valley,
                          Through two luminous windows, saw
                    Spirits moving musically,
                          To a lute's well-tuned law,
                    Round about a throne where, sitting
                          (Porphyrogene !)
                    In state his glory well befitting,
                          The ruler of the realm was seen.

                    And all with pearl and ruby glowing
                          Was the fair palace door,
                    Through which came flowing, flowing,
                          And sparkling evermore,
                    A troop of Echoes, whose sweet duty
                          Was but to sing,
                    In voices of surpassing beauty,
                          The wit and wisdom of their king.
\end{lstlisting}
\end{center}

\bigskip

Our hash function returns : 
\begin{alltt}
\noindent FF51DA4E7E50FBA7A8DC6858E9EC3353BDE2E465E1A6A1B03BEAA12A4AD694FB
\end{alltt}
\bigskip

If we put an additional space before ``       Was the fair palace door,'' the hash function
returns: 
\begin{alltt}
\noindent 03ABFA49B834D529669CFC1AEEC13E14EA5FFD2349582380BCBDBF8400017445
\end{alltt}
\bigskip

If we replace "Echoes" by "echoes" in the original text, the hash function
returns: 
\begin{alltt}
\noindent FE54777C52D373B7AED2EA5ACAD422B5B563BB3B91E8FCB48AAE9331DAC54A9B
\end{alltt}

\section{Quantitative measures}

\subsection{General definitions}

In the previous section we proved that discrete iterations produce a
topological chaos by checking two qualitative properties, namely
transitivity and regularity. This mathematical framework offers tools to
measure this chaos quantitatively.\newline
The first of this measures is the constant of sensitivity defined in
definition \ref{sensitivity}.\newline
Intuitively, a function $f$ has a constant of sensitivity equals to $\delta $
implies that there exist points aribitrarily close to any point $x$ which
eventually separate from $x$ by at least $\delta $ under some iterations of $%
f$.\newline
This induces that an arbitrarily small error on a the initial condition may
become magnified upon iterations of $f.$ (This is related to the famous 
\emph{butterfly effect}).

Other important tools are defined below.

\begin{definition}
A function $f$ is said to have the property of \emph{expansivity} if 
\begin{equation*}
\exists \varepsilon >0,\forall x\neq y,\exists n\in \mathbb{N}%
,d(f^{n}(x),f^{n}(y))\geqslant \varepsilon .
\end{equation*}

\noindent Then, $\varepsilon $ is the \emph{constant of expansivity }of $f.$ We also say $f$ is $\varepsilon$-expansive.
\end{definition}

\begin{remark}
A function $f$ has a constant of expansivity equals to $\varepsilon $ if an
arbitrary small error on any initial condition is amplified till $%
\varepsilon $.
\end{remark}

There exist other important quantitative tools such as topological entropy,
which quantifies the information contained at each iteration. But this is
not in the objective of this paper.\newline
We will reconsider this quantitative measures in the next subsection, in
relation with hash functions.

\subsection{Quantitative evaluation of our hash function}

Let $f_{0}$ be the vectorial logical negation previously used in our
algorithm. In this section, sensitivity and expansivity constants of $%
G_{f_{0}}$ will be calculated.

\subsubsection{Sensitivity}

We know that $(\mathcal{X},G_{f_{0}})$ has sensitive dependence on initial
conditions. Moreover, we have the following result.

\begin{theorem}
The constant of sensitivity of $(\mathcal{X},G_{f_{0}})$ is equal to $%
\mathsf{N}$.
\end{theorem}

Recall that $\mathsf{N}=256$ in our hash function.

\begin{proof}
We have seen that sensitivity is a consequence of having Devaney's chaos
property. Let us determine its constant.\newline
Let $(S,E)$ be a point of $\mathcal{X}$, and $\delta >0$. Then, let us
define another point $(S^{\prime },E^{\prime })$ by:

\begin{itemize}
\item $E^{\prime}=E$,

\item The $k$-th first terms of $S^{\prime }$ are the same as those of $S$,
where $k=\lfloor log_{10}(\varepsilon )\rfloor +1$ such that 
\begin{equation*}
d((S,E);(S^{\prime },E^{\prime }))<\delta .
\end{equation*}

\item Then, we put the terms $1, 2, 3, \hdots, \mathsf{N}$ to $S^{\prime}$.

\item $S^{\prime }$ can be completed by any terms.
\end{itemize}

\noindent Then it can be found a point $(S^{\prime },E^{\prime })$ closed to 
$(S,E)$ $\left( d((S,E);(S^{\prime },E^{\prime }))<\delta \right) $, such
that states of $G_{f_{0}}^{k+\mathsf{N}}(S,E)$ and $G_{f_{0}}^{k+\mathsf{N}%
}(S^{\prime },E^{\prime })$ differ for each cell, so that the distance
between this two points is greater or equal to $\mathsf{N}$. This proves
that we have sensitive dependence on the original text and that the constant
of sensitivity is $N.$
\end{proof}

\subsubsection{Expansivity}

\begin{theorem}
$(\mathcal{X},G_{f_{0}})$ is an expansive chaotic system. Its constant of
expansivity is equal to 1.
\end{theorem}

\begin{proof}
If $(S,E)\neq (\check{S};\check{E})$, then:

\begin{itemize}
\item Either $E\neq \check{E}$, and then at least one cell is not in the
same state in $E$ and $\check{E}$. Then the distance between $(S,E)$ and $(%
\check{S};\check{E})$ is greater or equal to 1.

\item Or $E=\check{E}$. Then the strategies $S$ and $\check{S}$ are not
equal. Let $n_{0}$ be the first index in which the terms $S$ and $\check{S}$
differ. Then 
\begin{equation*}
\forall k<n_{0},G_{f_{0}}^{n_{0}}(S,E)=G_{f_{0}}^{k}(\check{S},\check{E}),
\end{equation*}%
and $G_{f_{0}}^{n_{0}}(S,E)\neq G_{f_{0}}^{n_{0}}(\check{S},\check{E})$,
then as $E=\check{E},$ the cell which has changed in $E$ at the $n_{0}$-th
iterate is not the same than the cell which has changed in $\check{E}$, so
the distance between $G_{f_{0}}^{n_{0}}(S,E)$ and $G_{f_{0}}^{n_{0}}(\check{S%
},\check{E})$ is greater or equal to 2.
\end{itemize}
\end{proof}

The property of expansivity is a kind of avalanche effect.

Remark that it can be easily proved that $(\mathcal{X},G_{f_0})$ is not $A$%
-expansive, for any $A > 1$.

\section{Discussion and future work}

We proved that discrete chaotic iterations are a particular case of
Devaney's topological chaos if the iteration function is topologically
transitive and that the set of topologically transitive functions is non
void.\newline
We applied our results to the generation of new hash functions. Even if we
used the vectorial boolean negation function, our procedure remains general
and other transitive functions can be used.\newline
By considering hash functions as an application of our theory, we have shown
how some desirable aspects in encryption such as unpredictability,
sensitivity to initial conditions, mixture and disorder can be
mathematically guaranteed and even quantified by mathematical tools.

Theory of chaos recalls us that simple functions can have, when iterated, a
very complex behaviour, while some complicated functions could have
foreseeable iterations. This is why it is important to have tools for
evaluating desired properties.\newline
In our example, we used a simple topologically transitive iteration
function, but it can be proved that there exist a lot of functions of this
kind. Our simple function may be replaced by other "chaotic" functions which
can be evaluated with the above described quantitative tools. Another
important parameter is the choice of the strategy S. We proposed a
particular strategy that can be easily improved by multiple ways.%
\newline
We do not claim to have proposed a hash function replacing well known ones,
we simply wished to show how our mathematical context allows to build such
functions and especially how important properties can be measured.

Much work remains to be made, for example we are convinced that the good
comprehension of the transitivity property, enables to study the problem of
collisions in hash functions. \newline
In future work we plan to investigate other forms of chaos such as Li-York
chaos \cite{Li75} and to explore other quantitative and qualitative tools
such as entropy (see e.g. \cite{Bowen}) and to enlarge the domain of
applications of our theoretical concepts.

\bibliographystyle{plain}
\bibliography{mabase.bib}

\end{document}